\documentclass[conference]{IEEEtran}

\usepackage{amsmath,amssymb,amsfonts,amsthm}
\usepackage{color}
\usepackage{graphicx}
\usepackage{algpseudocode} 
\usepackage{algorithm}

\allowdisplaybreaks[4]

\newtheorem{lemma}{\textbf{Lemma}}

\newtheorem{thm}{\textbf{Theorem}}

\newtheorem{example}{\textbf{Example}}

\newcommand{\ie}{{i.e.}}

\newcommand{\mc}[1]{\mathcal{#1}}

\newcommand{\mbb}[1]{\mathbb{#1}}

\title{Multi-Message Private Information Retrieval\\with Private Side Information}

\author{\IEEEauthorblockN{Seyed Pooya Shariatpanahi}
	\IEEEauthorblockA{School of Computer Science \\Institute for Research  \\ in Fundamental Sciences (IPM) \\
		Email: pooya@ipm.ir}
	\and
	\IEEEauthorblockN{Mahdi Jafari Siavoshani}
	\IEEEauthorblockA{Department of Computer Engineering\\
		Sharif University of Technology \\
		Email:  mjafari@sharif.edu}
	\and
	\IEEEauthorblockN{Mohammad Ali Maddah-Ali}
	\IEEEauthorblockA{Nokia Bell Labs\\
		Holmdel, NJ, USA\\
		Email:  mohammad.maddahali@nokia-bell-labs.com}	
	
}

\begin{document}
\maketitle

\begin{abstract}
We consider the problem of private information retrieval (PIR) where a single user with private side information aims to retrieve multiple files from a library stored (uncoded) at a number of servers. We assume the side information at the user includes a subset of files stored privately (\ie, the server does not know the indices of these files). In addition,  we require that the identity of the requests and  side information at the user are not revealed to any of the servers. The problem involves finding the minimum load to be transmitted from the servers to the user such that the requested files can be decoded with the help of received and side information. By providing matching lower and upper bounds, for certain regimes, we characterize the minimum load imposed to all the servers (\ie, the capacity of this PIR problem). Our result shows that the capacity is the same as the capacity of a multi-message PIR problem without private side information, but with a library of reduced size. The effective size of the library is equal to the original library size minus the size of side information.
\end{abstract}


\section{Introduction}
Private Information Retrieval (PIR) is the problem of downloading contents from a library stored at a number of servers, without revealing the indices of requested contents to any of the servers. This problem was first considered in the computer science society, from a computational complexity viewpoint \cite{Chor1998}, \cite{Gasarch04asurvey}, which resulted in elegant cryptographic PIR schemes even for the simple scenario of a single user connected to a single server. Recently, in \cite{SunJafar2016_PIR_Capacity},  the capacity of this problem has been characterized from an information theoretic perspective\footnote{This notion of information theoretic privacy is much stronger than the notion of cryptographic privacy used in the computer science society.}. 
Here, the \emph{capacity} refers to the supremum of the number of decoded bits per each downloaded bit.
The setup considered in \cite{SunJafar2016_PIR_Capacity} consists of multiple non-colluding databases (servers) having access to a common library, from which a single user requests one file, via a shared link. Interestingly, as shown in \cite{SunJafar2016_PIR_Capacity}, the user can wisely send queries to the non-cooperative servers so that the aggregate load imposed to servers is minimized, while no server gets any information about the requested file index. Following the success of \cite{SunJafar2016_PIR_Capacity} in characterizing the exact capacity of the problem, many works have investigated extensions of this setup, such as coded databases \cite{Tajeddine2018,Banawan2018},  colluding databases \cite{SunCollude2018}, multi-message PIR \cite{Banawan2017}, adversarial PIR \cite{Wang2017}, PIR with asymmetric traffic at servers \cite{BanawanAsymm2018}, and private function retrieval \cite{SunJafar_ArXiv_CPC, MirmohseniMaddah_ArXiv_PFR}.

An important extension to the PIR problem is the case when the user has access to some form of stored data. 
This data can be provided by caching data during network off-peak hours (known as the \emph{Cache Content Placement} phase), to reduce the required load of servers when the actual request arrives (known as the \emph{Content Delivery} Phase). 
The first work considering this line is \cite{Tandon17}, which characterizes the capacity of a public cache setup for a multi-server scenario\footnote{By \emph{public} we mean that all the servers know the cache contents.}. In contrast, \cite{WeiUlukus_PIR_Cache_ArXiv17} considers a PIR problem with private cache where the servers do not know the cache contents.
Moreover, \cite{Wei2017} considers a scenario where the cache content placement is done via the same servers used in the content delivery phase, hence resulting in a partially known cache scenario, \ie, each server only knows which contents itself has sent during the cache content placement phase.

While in the cache aided PIR problem one can design the cache contents, in the PIR problem with side information it is assumed that the user has access to a given subset of the library files.
Along this direction, the authors in \cite{Kadhe2017} consider a single server scenario with private side information, and characterize the capacity by reducing the PIR setup to an index coding problem. The work \cite{ChenJafar2017_PIR_PSI} extends this result to the original multiple server setup with a user having a private side information. 

In this paper, we consider a new PIR scenario, namely multi-message PIR with private side information. In this setup we assume that the user requires multiple messages and at the same time has access to a private side information, which is a subset of the files in the library.
This generalized setup includes the scenarios considered in \cite{Banawan2017} (i.e., multi-message PIR without any side information) and \cite{ChenJafar2017_PIR_PSI} (i.e., single-message PIR with private side information) as its special cases. We propose lower and upper bounds on the required load, which match in some regimes of problem parameters, and thus, characterize the PIR capacity of those regimes. 
Specifically, we establish that if the user has access to a private side information of size $M$ files, then the capacity of the problem is the same as the capacity of the multi-message PIR problem (see \cite{Banawan2017}) without any side information, but with a library of size reduced by $M$. This result is a generalization of the same finding in \cite{ChenJafar2017_PIR_PSI} for the single-message PIR setting, and suggests the following conjecture: \emph{The capacity of any PIR problem with private side information of size $M$ is the same as capacity of the same problem without side information, but with a library size reduced by $M$}.

The structure of paper is as follows. In Section \ref{Sec:ProbSetup} we describe the problem setup. In Section \ref{Sec:MainRes} our main result  for the capacity of a multi-message PIR problem with private side information is stated. Sections \ref{Sec:Achivability} and \ref{Sec:Converse} provide the achievability and converse proofs, respectively. Finally, Section \ref{Sec:Conclude} concludes the paper.

\section{Problem Setup}\label{Sec:ProbSetup}
Consider a single user connected to $N$ servers having access to a library of $K$ files $\{W_1,\ldots,W_K\}$, where each file consists of $L$ symbols chosen independently and uniformly at random from a finite field $\mathbb{F}$, \ie, $W_i = (w_i(1),\ldots,w_i(L))$. We assume the user has a private side information which contains $M$ files from the library,  denoted by $W_\mathcal{S} \triangleq \{W_i: i \in \mathcal{S} \}$, $\mc{S}\subseteq [1:K], |\mathcal{S}|=M$, where the servers do not know the index set $\mathcal{S}$. The user wishes to retrieve $P$ new file $W_{\mathcal{P}}\triangleq \{W_i: i \in \mathcal{P}\}$, $\mc{P}\subseteq [1:K], |\mathcal{P}|=P$, where $\mathcal{P} \cap \mathcal{S} = \varnothing$. To this end, the user sends a set of queries $Q^{\mc{S},\mc{P}} \triangleq \{Q^{\mathcal{S},\mathcal{P}}_n, n \in [1:N]\}$, where server $n$ just receives $Q^{\mathcal{S},\mathcal{P}}_n$ without having any access to other queries (this is known as the non-colluding servers assumption).  The queries should be designed such that the servers do not obtain any information about neither the requested file index set $\mathcal{P}$ nor the side information index set $\mathcal{S}$ as formally stated in the following
\[
I(Q^{\mc{S},\mc{P}}_n ; \mc{P}, \mc{S}) = 0, \quad \forall n\in [1:N],  \quad \text{\emph{(privacy constraint)}}.
\]
After receiving the queries, each server $n \in [1:N]$ sends the answer $A^{\mc{S},\mc{P}}_n$ which is a function of library contents $W_{1:K}$ and the query $Q^{\mathcal{S},\mc{P}}_n$ received at that server. The answers must be designed such that there exists a decoding function $\Psi$ which satisfies
\[
\Psi(A^{\mc{S},\mc{P}}_{[1:N]}, Q^{\mc{S},\mc{P}}_{[1:N]}, W_\mc{S}) = W_\mc{P}, \quad \text{\emph{(decodability constraint)}}.
\]

The objective is to characterize the minimum required download symbols defined as follows
\[
D^{\mathsf{PSI}}(N,K,P,M) \triangleq \inf \frac{1}{PL} \sum_{n=1}^{N} H(A^{\mc{S},\mc{P}}_n),
\]
so that the \emph{privacy} and \emph{decodability} constraints are satisfied, where the infimum is taken over all possible strategies. Equivalently, the PIR capacity can be defined as $C^{\mathsf{PSI}}(N,K,P,M) \triangleq \frac{1}{D^{\mathsf{PSI}}(N,K,P,M)}$.

Note that in the two special cases of $P=1$ and $M=0$ our setup reduces to those studied in  \cite{ChenJafar2017_PIR_PSI} and \cite{Banawan2017}, respectively.

\section{Main Result}\label{Sec:MainRes}
In this paper, we characterize the capacity of a multi-message PIR problem with side information. Our main result is formally stated in Theorem~\ref{thm:MainResult}.

\begin{thm}\label{thm:MainResult}
For $P \geq \frac{K-M}{2}$ we have 
\begin{equation}\label{Eqqqqq}
D^{\mathsf{PSI}}(N,K,P,M) = 1 + \frac{K-M-P}{PN}.
\end{equation}
Moreover, for $P \leq \frac{K-M}{2}$ and $\frac{K-M}{P} \in \mathbb{N}$ we have 
\begin{equation}
D^{\mathsf{PSI}}(N,K,P,M) = \frac{1- \left(\frac{1}{N}\right)^{(K-M)/P}}{1- \left(\frac{1}{N}\right)}.
\end{equation}
\end{thm}

It is worth mentioning that Theorem~\ref{thm:MainResult} can be equivalently stated as
\[
D^{\mathsf{PSI}}(N,K,P,M) = D^{\mathsf{PSI}}(N,K-M,P,0),
\]
under the assumptions of theorem, where $D^{\mathsf{PSI}}(N,K-M,P,0)$ is characterized in \cite{Banawan2017}. This implies that for the multi-message PIR problem, introducing private side information of size $M$ into the problem setup reduces the problem's capacity to the capacity of a multi-message PIR problem without side information but with a library of size $K-M$.
Note that our result generalizes the same finding reported in \cite{ChenJafar2017_PIR_PSI} for the single message PIR problem with side information.

The above library size reduction effect would be trivial if the privacy constraint was not required for side information. However, we prove that the same performance can be achieved even with preserving the privacy constraint for side information.
The main ingredient of the proposed achievable scheme is to use an outer layer of an MDS code to leverage the private side information available at the user.

\section{The Achievability Proof}\label{Sec:Achivability}
Let us first start with a motivating example for the case of $P\ge \frac{K-M}{2}$. 
\begin{example}
Suppose we have $N=2$ servers, $K=4$ files, $M=1$ file as side information, and the user requests $P=2$ new files. First, we review the achievable scheme proposed in \cite{Banawan2017} which is designed for the $M=0$ case, \ie, no side information. In their scheme, each file contents is permuted independently and uniformly at random, resulting in permuted files $A$, $B$, $C$, and $D$. Moreover, it is assumed that these permutations are not known by the servers. Next, each of the files is partitioned into $N^2=4$ equal-sized chunks, \ie, $A=(a_1,\ldots,a_4)$, $B=(b_1,\ldots,b_4)$, $C=(c_1,\ldots,c_4)$, and $D=(d_1,\ldots,d_4)$.

Their scheme consists of two phases. Suppose the user requires to privately retrieve files $A$ and $B$. 
In the first phase, each server sends a chunk from each file, \ie, Server~1 sends $a_1, b_1, c_1, d_1$ and Server~2 sends $a_2, b_2, c_2, d_2$.
In the second phase, the scheme uses a Reed-Solomon generator matrix in $\mathbb{F}_q$ as follows
\[
G_{2\times 3} = \left[ \begin{matrix}
1 & 1 & 1 & 1 \\
1 & 2 & 3 & 4
\end{matrix} \right].
\]
The user generates two new matrices $G_1$ and $G_2$ by applying two independent random permutations on the columns of $G$. Then, the user requests from the first and second servers
\[
G_1 \left[ \begin{matrix} a_3\ b_3\ c_2\ d_2 \end{matrix} \right] ^\mathsf{T} \quad\quad\text{and}\quad\quad
G_2 \left[ \begin{matrix} a_4\ b_4\ c_1\ d_1 \end{matrix} \right]^\mathsf{T},
\]
respectively. Since user has received $c_2$ and $d_2$ from the second server in the first phase, it can decode $a_3$ and $b_3$ from the above two linear equations sent by the first server in the second phase. Similarly, it can decode $a_4$ and $b_4$. Thus, the user has retrieved files $A$ and $B$. This scheme results in the load of $D^{\mathsf{PSI}}(N=2,K=4,P=2,M=0)=3/2$ per decoded file.

Now, suppose the user has access to a private side information of size $M=1$. We use the fact that the user can construct some of the above transmissions from its side information directly. 
In the above, Server~1 sends $6$ coded chunks where one of the file chunks in the first phase is already available as side information, which can be used to reduce the load of the first server.
Since the side information should remain private to the server, we employ an MDS generator matrix $G_{5\times 6}$. Then Server~1 sends $5$ linear equations of the original $6$ coded chunks with coefficients obtained from the rows of $G_{5\times 6}$.
By removing the term already available as side information, the user can form a linear system of $5$ equations to decode the remaining chunks.
The same procedure is followed by the second server. This results in the load of $D^{\mathsf{PSI}}(N=2,K=4,P=2,M=1)=5/4$ per decoded file. Notice that as stated in Theorem~\ref{thm:MainResult}, this scheme achieves the optimal load.

 \hfill $\square$
\end{example}
Although the above example was stated for the case $P\ge \frac{K-M}{2}$, since the main idea of introducing the role of side information into the achievable scheme for the other case of $P\le \frac{K-M}{2}$ is the same, we skip presenting an example for this case. In the following, we explain the general achievable scheme.

The proposed scheme is along the same line introduced in \cite{ChenJafar2017_PIR_PSI}. Let us first review the concept of linear PIR schemes. 
Suppose the user chooses $\pi_i$, $i \in [1:K]$, independently and uniformly at random from the set of all permutations of $[1:L]$, hidden from the servers.
We define the scrambled version of file $W_i$ as $U_i=(u_i(1),\ldots,u_i(L))=(w_i(\pi_i(1)),\ldots,w_i(\pi_i(L)))$. 
Then, we group these permuted symbols into chunks of $c$ symbols, \ie, $U_i = (v_i(1),\ldots,v_i(\frac{L}{c}))$ where $v_i(j) = \big( u_i( (j-1)c+1),\ldots, u_i(jc) \big)$ for $j\in[1: L/c]$, in which for the sake of presentation clarity we have assumed that $c$ divides $L$.

A $k$-sum of type $(j_1,\ldots,j_k)$ with coefficients vector $(\delta_1,\ldots,\delta_k)$ of chunks is defined as 
$\delta_1 v_{j_1}(i_1)+\cdots+\delta_k v_{j_k}(i_k)$ 
where $j_1,\ldots,j_k$ are distinct elements of $[1:K]$, $i_1,\ldots,i_k\in [1:L/c]$, $\delta_i\in\mbb{F}$, and all the operations are performed element-wise in $\mbb{F}$.

In a general linear PIR scheme, each server transmits $K$ blocks where the block $k\in[1:K]$ consists of all possible types of $k$-sums. This symmetric structure is a consequence of the privacy requirement \cite{SunJafar2016_PIR_Capacity}. Moreover, each type of $k$-sums appears in $\alpha_{N,K}(k)$ distinct instances, by involving different chunks from the corresponding files in the $k$-sum. Also, each distinct instance mentioned above appears $\beta_{N,K}(k)$ times with possibly different coefficient vectors. Thus in total, each block consists of $\binom{K}{k} \alpha_{N,K}(k) \beta_{N,K}(k)$ chunks which results in the total load of 
\begin{equation}\label{eq:p_def}
p(N,K) \triangleq c_{N,K} \sum_{k=1}^K \binom{K}{k} \alpha_{N,K}(k) \beta_{N,K}(k)
\end{equation}
symbols, imposed to each server. Notice that in a general achievable scheme, the chunk size $c$ depends on $N$ and $K$, hence we denote it by $c_{N,K}$.
A linear PIR achievable scheme is defined to be \emph{valid} if it satisfies both the decodability and privacy constraints, defined in Section~\ref{Sec:ProbSetup}.

Now, for the general PIR problem with private side information, we describe an achievable scheme by employing MDS codes on top of the PIR scheme without private side information, similar to the special case proposed in \cite{ChenJafar2017_PIR_PSI}. In order to do this, assume $p(N,K)$ symbols are transmitted from each server in the PIR scheme without side information.  Then, if the user is equipped with side information, it is clear that a subset of these symbols can be constructed directly from the side information, and thus should not be transmitted. Let us denote the number of such symbols by $q(N,K,M)$, which can be calculated as follows
\begin{equation}\label{eq:q_def}
q(N,K,M) \triangleq  c_{N,K} \sum_{k=1}^M \binom{M}{k} \alpha_{N,K}(k) \beta_{N,K}(k).
\end{equation}
Now, consider an MDS code\footnote{In the following, we may remove the dependency of functions $p$ and $q$ on $N$, $K$, and $M$ wherever it is clear from the context.} $[2p-q,p]$ over the finite field $\mathbb{F}$ where each server encodes its $p$ symbols with such an MDS code. Then, instead of sending the original $p$ symbols, each server only transmits $p-q$ symbols corresponding to the non-systematic part of such a code. Accordingly, the total number of symbols, per decoded file, required to be transmitted by all servers, in the presence of private side information, is
\[
D^{\mathsf{PSI}}(N,K,P,M) = \frac{N}{PL} \Big[ p(N,K)-q(N,K,M) \Big].
\]
It can be easily verified that the new scheme is valid if the original scheme is valid. That is because, first, the user can recover the remaining $p-q$ elements from the coded symbols and the side information due to the MDS code properties. Second, the new scheme sends the same set of queries as the original scheme which satisfy the privacy constraint. Hence, the total load of new scheme is $p-q$ symbols per server.

The above construction is also used in \cite{ChenJafar2017_PIR_PSI} on top of the PIR problem studied in \cite{SunJafar2016_PIR_Capacity}. The interesting finding of \cite{ChenJafar2017_PIR_PSI} is that the load of the PIR problem with private side information of size $M$ is equal to the load of a PIR problem without private side information where the library size is reduced by $M$, \ie, $K-M$.

In this section, we use the general construction introduced above to develop an achievable scheme for the multi-message PIR scheme with private side information, based on the original scheme proposed in \cite{Banawan2017}. Thus, the main question to be answered is whether in this case providing the user with private side information of size $M$ will reduce the effective library size from $K$ to $K-M$. 

The answer to the above question is a \emph{Yes} if the following constraint is satisfied
\begin{align}\label{eq:Achv_Equality_Constraint}
p(N,K)-q(N,K,M) = p(N,K-M),
\end{align}
where $p$ and $q$, defined in \eqref{eq:p_def} and \eqref{eq:q_def}, are determined by the achievable scheme through the coefficients $\alpha_{N,K}(k)$,  $\beta_{N,K}(k)$ and $c_{N,K}$. Thus, it just remains to check if the corresponding coefficients used in the achievable scheme proposed for the multi-message problem without private side information in \cite{Banawan2017} satisfy \eqref{eq:Achv_Equality_Constraint} or not. In the rest of this section we answer this question in the affirmative.

\subsection{Analysis of the Achievable Scheme for $P \geq \frac{K-M}{2}$}
For this regime, considering the scheme introduced in  \cite{Banawan2017} and without going into the details, one can verify the followings
\begin{align*}
   \alpha_{N,K}(k) &= \left\{
     \begin{array}{lcl}
       1  &:&  k=1,\\
       0  &:&  k=2,\dots,K-1, \\
       N-1  &:&  k=K,
     \end{array}
   \right.
\end{align*}  

\begin{align*} 
   \beta_{N,K}(k) &= \left\{
        \begin{array}{lcl}
          1  &:&  k=1,\\
          0  &:&  k=2,\dots,K-1, \\
          P  &:&  k=K,
        \end{array}
      \right.
\end{align*}
and
$c_{N,K} = \frac{L}{N^2}$. 
Then, we  have
\begin{align*}
\hspace{0pt} p(N,K)-q(N,K,M) 
 &= c_{N,K} \sum_{k=1}^K \binom{K}{k} \alpha_{N,K}(k) \beta_{N,K}(k) \\
&\quad -  c_{N,K}\sum_{k=1}^M \binom{M}{k} \alpha_{N,K}(k) \beta_{N,K}(k) \\ 
&= \frac{L}{N^2} \Big( K-M+P(N-1) \Big).
\end{align*}
On the other hand we have
\begin{align*}
p(N,K-M) & \\
&\hspace{-40pt} = c_{N,K-M} \sum_{k=1}^{K-M} \binom{K-M}{k} \alpha_{N,K-M}(k) \beta_{N,K-M}(k) \\ \nonumber
&\hspace{-40pt} = \frac{L}{N^2} \Big( K-M+P(N-1) \Big)
\end{align*}
which confirms \eqref{eq:Achv_Equality_Constraint}, and concludes the proof.

\subsection{Analysis of the Achievable Scheme for $P \leq \frac{K-M}{2}$}
By inspecting the scheme introduced in \cite{Banawan2017} for the case of $P \leq \frac{K-M}{2}$, we can verify that
\begin{align*}
\alpha_{N,K}(k) &= \sum_{i=1}^P \gamma^{N,K}_i r_i^{K-P-k} 
= r^{K-P-k} \sum_{i=1}^P \gamma^{N,K}_i,
\end{align*}
where according to \cite{Banawan2017} we have $r_i=r$ and  $r$ satisfies $(1+\frac{1}{r})^K = N^{K/P }$. Here, $(\gamma^{N,K}_1, \ldots,\gamma^{N,K}_P)$ is the solution of the linear equation \cite[Eq.~(64)]{Banawan2017}. Notice that we do not need the exact values of $\gamma^{N,K}_i$'s since they will cancel out later from our computations. Moreover, we have
$\beta_{N,K}(k) = 1$ in this regime.

The constant $c_{N,K}$ of the scheme in \cite{Banawan2017} can be derived as follows. Each $k$-sum transmitted by each server can be one of the two following types. The first type includes those sums containing at least some chunks from the requested files and maybe some chunks from the not-requested files. The second type only consists of sums containing not-requested file chunks, which are being transmitted due to the privacy requirements. For the decodability constraint we require that the total number of useful equations (equations containing at least one chunk from the $P$ requested files) transmitted by all the servers should be equal to the total number of symbols required to recover the $P$ requested files (i.e., $PL$ symbols). This will result in
\begin{align*}
c_{N,K} = \frac{1}{N}\frac{PL}{ \sum_{k=1}^K \binom{K}{k} \alpha_{N,K}(k) - \sum_{k=1}^{K-P}  \binom{K-P}{k} \alpha_{N,K}(k) }.
\end{align*}
Hence, to verify \eqref{eq:Achv_Equality_Constraint} we can proceed as follows
\begin{align*}
&\hspace{-15pt} p(N,K) - q(N,K,M) \\
&= \frac{PL}{N}\frac{ \sum_{k=1}^K \binom{K}{k}\alpha_{N,K}(k) - \sum_{k=1}^M \binom{M}{k} \alpha_{N,K}(k) }{ \sum_{k=1}^K \binom{K}{k} \alpha_{N,K}(k) - \sum_{k=1}^{K-P}  \binom{K-P}{k} \alpha_{N,K}(k) } \\
&= \frac{PL}{N} \frac{ \sum_{k=1}^K \binom{K}{k} r^{K-P-k} -\sum_{k=1}^M \binom{M}{k} r^{K-P-k} }{ \sum_{k=1}^K \binom{K}{k} r^{K-P-k} - \sum_{k=1}^{K-P}  \binom{K-P}{k} r^{K-P-k} }\\
&= \frac{PL}{N} \frac{ \sum_{k=1}^K \binom{K}{k} r^{-k} -\sum_{k=1}^M \binom{M}{k} r^{-k} }{ \sum_{k=1}^K \binom{K}{k} r^{-k} - \sum_{k=1}^{K-P}  \binom{K-P}{k} r^{-k} }\\
&= \frac{PL}{N} \frac{ (1+r^{-1})^K - (1+r^{-1})^M }{ (1+r^{-1})^K - (1+r^{-1})^{K-P} } \\
&= \frac{PL}{N} \frac{ N^{K/P} - N^{M/P} }{ N^{K/P} - N^{(K-P)/P}  } \\
&= \frac{PL}{N} \frac{1- \left(\frac{1}{N}\right)^{\frac{K-M}{P}} }{ 1 - \frac{1}{N}}\\
&= \frac{PL}{N} \cdot D^{\mathsf{PSI}}(N,K-M,P,0).
\end{align*}
which is equal to $p(N,K-M)$ according to \cite{Banawan2017}. This concludes the proof.

\section{The Converse Proof}\label{Sec:Converse}
\subsection{The Converse Argument for $P \geq \frac{K-M}{2}$}
To prove the converse for this case, let us proceed as follows.
Without loss of generality we assume $\mc{S}=\{1,\ldots,M\}$, $\mc{P}=\{M+1,\ldots,M+P+1\}$. 
Then, we can state the following lemma.
\begin{lemma}\label{lem:ConvHighP_lem1}
	Under the assumptions of Theorem~\ref{thm:MainResult}, we have
	\begin{align}
	(K-M)L &\le  \sum_{n=1}^N H(A^{\mc{S},\mc{P}}_n | Q, W_\mc{S}) + (K-M-P)L \notag\\
	&\quad - H(A_1 | Q, W_\mc{S}, W_\mc{P}).
	\end{align}
\end{lemma}
\begin{proof}
	The proof is provided in the Appendix.
\end{proof}

Now, we can lower bound $D^{\mathsf{PSI}}$ as follows
\begin{align}
D^{\mathsf{PSI}} &= \frac{1}{PL}\sum_{n=1}^{N} H(A^{\mc{S},\mc{P}}_n) \notag\\
&\ge  \frac{1}{PL}\sum_{n=1}^{N} H(A^{\mc{S},\mc{P}}_n | Q, W_{\mc{S}}) \notag\\
&\stackrel{\text{(a)}}{\ge}  \frac{1}{PL} \Big(PL + H(A_1 | Q, W_{\mc{S}}, W{\mc{P}}) \Big) \notag\\
&\stackrel{\text{(b)}}{\ge} 1 + \frac{K-M-P}{NP} \notag
\end{align}
where (a) follows from Lemma~\ref{lem:ConvHighP_lem1} and (b) holds due the following lemma, Lemma~\ref{lem:InterferenceLowerBound}.
This concludes the proof.

\begin{lemma}\label{lem:InterferenceLowerBound}
If $P\ge \frac{K-M}{2}$, then we have
\[
H(A_1 | Q, W_{\mc{S}}, W_{\mc{P}}) \ge \frac{(K-M-P)L}{N}.
\]
\end{lemma}
\begin{proof}
The proof is provided in the Appendix.
\end{proof}

\subsection{The Converse Argument for $P \leq \frac{K-M}{2}$}
To prove the converse for this regime, we can state the following lemma.

\begin{lemma}\label{lem:ConvLowP_lem1}
	Under the assumption of Theorem~\ref{thm:MainResult}, we have
	\begin{align*}
	(K-M)L &\le  N H(A_1 | Q, W_\mc{S}) \\
	&\quad+ (N-1)[NH(A_1 | Q, W_\mc{S}) - PL] \\
	&\quad + (K-M-2P)L \\
	&\quad - H(A_1 | W_{[M+1:M+2P]}, Q, W_\mc{S} ).
	\end{align*}
\end{lemma}
\begin{proof}
The	proof is provided in the Appendix
\end{proof}
The result stated in Lemma~\ref{lem:ConvLowP_lem1} can be equivalently written as
\begin{align}\label{eq:ConvLowP_InductionRelation}
N H(A_1 | Q, W_\mc{S}) &\ge \left(1+\frac{1}{N} \right)PL \notag\\
&\quad + \frac{1}{N} H(A_1 | W_{[M+1:M+2P]},Q,W_\mc{S}).
\end{align}

Equation~\eqref{eq:ConvLowP_InductionRelation} provides a lower bound on $H(A_1 | Q, W_\mc{S})$. Applying the same chain of inequalities used to prove Lemma~\ref{lem:ConvLowP_lem1}, one can derive a similar inequality for $H(A_1 | W_{[M+1:M+2P]},Q,W_\mc{S})$. Hence, by applying \eqref{eq:ConvLowP_InductionRelation} inductively over the library size, one can derive a lower bound for $H(A_1 | Q, W_\mc{S})$, similar to \cite[Equation~(125)]{Banawan2017}.
The important difference of above calculations with those appeared in \cite{Banawan2017} is that here we start the induction from the message $W_{M+1}$ while the starting point of \cite{Banawan2017} is the message $W_1$.

Similar to the path followed from \cite[Equation~(124)]{Banawan2017} to \cite[Equation~(133)]{Banawan2017}, we arrive at 
\begin{align}
D^{\mathsf{PSI}} \ge \left( \frac{1- \left(\frac{1}{N}\right)^{\left\lfloor \frac{K-M}{P} \right\rfloor} }{1-\frac{1}{N}}  + \frac{\left( \frac{K-M}{P} - \left\lfloor \frac{K-M}{P} \right\rfloor \right)}{N^{\left\lfloor \frac{K-M}{P} \right\rfloor}}   \right) \notag
\end{align}
which completes the proof for the case of $P\le \frac{K-M}{2}$ and $\frac{K-M}{P}\in \mathbb{N}$.

\section{Conclusions}\label{Sec:Conclude}
In this paper, we have characterized the capacity of a multi-message PIR problem where the user has access to a private side information, for certain regimes. Our result shows that the role of private side information is equivalent to reducing the effective library size by the size of side information. This result, along with a similar conclusion for a single-message PIR setup in \cite{ChenJafar2017_PIR_PSI}, motivate the conjecture that for any PIR scheme, adding a private side information will be equivalent to reducing the effective library size.

\bibliographystyle{IEEEtran}

\bibliography{pir}

\begin{thebibliography}{10}
\providecommand{\url}[1]{#1}
\csname url@samestyle\endcsname
\providecommand{\newblock}{\relax}
\providecommand{\bibinfo}[2]{#2}
\providecommand{\BIBentrySTDinterwordspacing}{\spaceskip=0pt\relax}
\providecommand{\BIBentryALTinterwordstretchfactor}{4}
\providecommand{\BIBentryALTinterwordspacing}{\spaceskip=\fontdimen2\font plus
\BIBentryALTinterwordstretchfactor\fontdimen3\font minus
  \fontdimen4\font\relax}
\providecommand{\BIBforeignlanguage}[2]{{%
\expandafter\ifx\csname l@#1\endcsname\relax
\typeout{** WARNING: IEEEtran.bst: No hyphenation pattern has been}%
\typeout{** loaded for the language `#1'. Using the pattern for}%
\typeout{** the default language instead.}%
\else
\language=\csname l@#1\endcsname
\fi
#2}}
\providecommand{\BIBdecl}{\relax}
\BIBdecl

\bibitem{Chor1998}
B.~Chor, E.~Kushilevitz, O.~Goldreich, and M.~Sudan, ``Private information
  retrieval,'' \emph{J. ACM}, vol.~45, no.~6, pp. 965--981, Nov. 1998.

\bibitem{Gasarch04asurvey}
W.~Gasarch, ``A survey on private information retrieval,'' \emph{Bulletin of
  the EATCS}, vol.~82, pp. 72--107, 2004.

\bibitem{SunJafar2016_PIR_Capacity}
H.~Sun and S.~A. Jafar, ``The capacity of private information retrieval,''
  \emph{IEEE Transactions on Information Theory}, vol.~63, no.~7, pp.
  4075--4088, July 2017.

\bibitem{Tajeddine2018}
R.~Tajeddine, O.~W. Gnilke, and S.~E. Rouayheb, ``Private information retrieval
  from mds coded data in distributed storage systems,'' \emph{IEEE Transactions
  on Information Theory}, pp. 1--1, 2018.

\bibitem{Banawan2018}
K.~Banawan and S.~Ulukus, ``The capacity of private information retrieval from
  coded databases,'' \emph{IEEE Transactions on Information Theory}, vol.~64,
  no.~3, pp. 1945--1956, March 2018.

\bibitem{SunCollude2018}
H.~Sun and S.~A. Jafar, ``The capacity of robust private information retrieval
  with colluding databases,'' \emph{IEEE Transactions on Information Theory},
  vol.~64, no.~4, pp. 2361--2370, April 2018.

\bibitem{Banawan2017}
\BIBentryALTinterwordspacing
K.~A. Banawan and S.~Ulukus, ``Multi-message private information retrieval:
  Capacity results and near-optimal schemes,'' \emph{CoRR}, vol.
  abs/1702.01739, 2017. [Online]. Available:
  \url{http://arxiv.org/abs/1702.01739}
\BIBentrySTDinterwordspacing

\bibitem{Wang2017}
\BIBentryALTinterwordspacing
Q.~Wang and M.~Skoglund, ``Secure symmetric private information retrieval from
  colluding databases with adversaries,'' \emph{CoRR}, vol. abs/1707.02152,
  2017. [Online]. Available: \url{http://arxiv.org/abs/1707.02152}
\BIBentrySTDinterwordspacing

\bibitem{BanawanAsymm2018}
\BIBentryALTinterwordspacing
K.~A. Banawan and S.~Ulukus, ``Asymmetry hurts: Private information retrieval
  under asymmetric traffic constraints,'' \emph{CoRR}, vol. abs/1801.03079,
  2018. [Online]. Available: \url{http://arxiv.org/abs/1801.03079}
\BIBentrySTDinterwordspacing

\bibitem{SunJafar_ArXiv_CPC}
\BIBentryALTinterwordspacing
H.~Sun and S.~A. Jafar, ``The capacity of private computation,'' \emph{CoRR},
  vol. abs/1710.11098, 2017. [Online]. Available:
  \url{http://arxiv.org/abs/1710.11098}
\BIBentrySTDinterwordspacing

\bibitem{MirmohseniMaddah_ArXiv_PFR}
\BIBentryALTinterwordspacing
M.~Mirmohseni and M.~A. Maddah{-}Ali, ``Private function retrieval,''
  \emph{CoRR}, vol. abs/1711.04677, 2017. [Online]. Available:
  \url{http://arxiv.org/abs/1711.04677}
\BIBentrySTDinterwordspacing

\bibitem{Tandon17}
\BIBentryALTinterwordspacing
R.~Tandon, ``The capacity of cache aided private information retrieval,''
  \emph{CoRR}, vol. abs/1706.07035, 2017. [Online]. Available:
  \url{http://arxiv.org/abs/1706.07035}
\BIBentrySTDinterwordspacing

\bibitem{WeiUlukus_PIR_Cache_ArXiv17}
\BIBentryALTinterwordspacing
Y.~Wei, K.~A. Banawan, and S.~Ulukus, ``Fundamental limits of cache-aided
  private information retrieval with unknown and uncoded prefetching,''
  \emph{CoRR}, vol. abs/1709.01056, 2017. [Online]. Available:
  \url{http://arxiv.org/abs/1709.01056}
\BIBentrySTDinterwordspacing

\bibitem{Wei2017}
\BIBentryALTinterwordspacing
------, ``Cache-aided private information retrieval with partially known
  uncoded prefetching: Fundamental limits,'' \emph{CoRR}, vol. abs/1712.07021,
  2017. [Online]. Available: \url{http://arxiv.org/abs/1712.07021}
\BIBentrySTDinterwordspacing

\bibitem{Kadhe2017}
\BIBentryALTinterwordspacing
S.~Kadhe, B.~Garcia, A.~Heidarzadeh, S.~Y.~E. Rouayheb, and A.~Sprintson,
  ``Private information retrieval with side information,'' \emph{CoRR}, vol.
  abs/1709.00112, 2017. [Online]. Available:
  \url{http://arxiv.org/abs/1709.00112}
\BIBentrySTDinterwordspacing

\bibitem{ChenJafar2017_PIR_PSI}
\BIBentryALTinterwordspacing
Z.~Chen, Z.~Wang, and S.~Jafar, ``The capacity of private information retrieval
  with private side information,'' \emph{CoRR}, vol. abs/1709.03022, 2017.
  [Online]. Available: \url{http://arxiv.org/abs/1709.03022}
\BIBentrySTDinterwordspacing

\end{thebibliography}

\newpage

\appendix
This appendix contains the omitted proofs from the main body of paper.

\begin{proof}[{Proof of Lemma~\ref{lem:ConvHighP_lem1}}]
To prove this lemma we proceed as follows
\begin{align}
&\hspace{-10pt} (K-M)L = H(W_{[1:K]\setminus\mc{S}}) \notag\\
&\stackrel{\text{(a)}}{=} H(W_{[1:K]\setminus\mc{S}} | Q, W_\mc{S}) \notag\\
&\stackrel{\text{(b)}}{=} H(W_{[1:K]\setminus\mc{S}} | Q, W_\mc{S}) \notag\\
&\quad - H(W_{[1:K]\setminus\mc{S}} | A^{\mc{S},\mc{P}_1}_{[1:N]}, \ldots, A^{\mc{S},\mc{P}_\beta}_{[1:N]},  Q, W_\mc{S}) \notag\\
&= I(W_{[1:K]\setminus\mc{S}} ; A^{\mc{S},\mc{P}_1}_{[1:N]}, \ldots, A^{\mc{S},\mc{P}_\beta}_{[1:N]} | Q, W_\mc{S}) \notag\\
&\stackrel{\text{(c)}}{=} H(A^{\mc{S},\mc{P}_1}_{[1:N]}, \ldots, A^{\mc{S},\mc{P}_\beta}_{[1:N]} | Q, W_\mc{S}) \notag\\
&\stackrel{}{=} H(A_1, A^{\mc{S},\mc{P}_1}_{[2:N]}, \ldots, A^{\mc{S},\mc{P}_\beta}_{[2:N]} | Q, W_\mc{S}) \label{eq:eq99}\\
&\stackrel{\text{(d)}}{=} H(A_1, A^{\mc{S},\mc{P}_1}_{[2:N]} | Q, W_\mc{S}) \notag\\
&\quad + H(A^{\mc{S},\mc{P}_2}_{[2:N]}, \ldots, A^{\mc{S},\mc{P}_\beta}_{[2:N]} | A_1, A^{\mc{S},\mc{P}_1}_{[2:N]}, Q, W_\mc{S}) \notag\\
&\stackrel{\text{(e)}}{=} H(A_1, A^{\mc{S},\mc{P}_1}_{[2:N]} | Q, W_\mc{S}) \notag\\
&\quad + H(A^{\mc{S},\mc{P}_2}_{[2:N]}, \ldots, A^{\mc{S},\mc{P}_\beta}_{[2:N]} | A_1, A^{\mc{S},\mc{P}_1}_{[2:N]}, Q, W_\mc{S}, W_\mc{P}) \notag\\
&\stackrel{\text{(f)}}{\le} \sum_{n=1}^N H(A^{\mc{S},\mc{P}}_n | Q, W_\mc{S}) \notag\\
&\quad + H(A^{\mc{S},\mc{P}_2}_{[2:N]}, \ldots, A^{\mc{S},\mc{P}_\beta}_{[2:N]} | A_1, Q, W_\mc{S}, W_\mc{P}) \notag\\
&\stackrel{\text{(g)}}{=} \sum_{n=1}^N H(A^{\mc{S},\mc{P}}_n | Q, W_\mc{S}) \notag\\
&\quad + H(A^{\mc{S},\mc{P}_2}_{[1:N]}, \ldots, A^{\mc{S},\mc{P}_\beta}_{[1:N]} | Q, W_\mc{S}, W_\mc{P}) \notag\\
&\quad - H(A_1 | Q, W_\mc{S}, W_\mc{P}) \notag\\
&\stackrel{\text{(h)}}{=} \sum_{n=1}^N H(A^{\mc{S},\mc{P}}_n | Q, W_\mc{S}) + (K-M-P)L \notag\\
&\quad - H(A_1 | Q, W_\mc{S}, W_\mc{P}) \label{eq:H_K-S_Simplified}
\end{align}
where $\beta \triangleq \binom{K-M}{P}$ and $\mc{P}_i$'s are all subsets of size $P$ from $[1:K]\setminus \mc{S}$. Notice that without loss of generality we can choose $\mc{P}_1 = \mc{P}$.
In the above chain of equations 
(a) holds since $W_{[1:K]\setminus\mc{S}}$ is independent from $Q$ and $W_\mc{S}$, (b) is true because $W_{[1:K]\setminus\mc{S}}$ can be decoded from $A^{\mc{S},\mc{P}_1}_{[1:N]}, \ldots, A^{\mc{S},\mc{P}_\beta}_{[1:N]}$,  $Q$, and  $W_\mc{S}$,
(c) follows since the answers are deterministic functions of the library and queries,
(d) follows from \cite[Lemma~1]{Banawan2017} (also, see \cite{SunJafar2016_PIR_Capacity}) and the chain rule,
(e) holds since $W_\mc{P}$ can be recovered from the side information $W_\mc{S}$ available at the user and answers $A^{\mc{S},\mc{P}}_{[1:N]}$,
(f) follows from \cite[Lemma~1]{Banawan2017} and the fact that conditioning does not increase entropy,
(g) chain rule
(h) holds since $W_{[1:K]\setminus(\mc{P}\cup\mc{S}) }$ can be recovered from $A^{\mc{S},\mc{P}_2}_{[1:N]}, \ldots, A^{\mc{S},\mc{P}_\beta}_{[1:N]}$, $W_{\mc{S}}$ and the queries.
\end{proof}

\begin{proof}[{Proof of Lemma~\ref{lem:InterferenceLowerBound}}]
	In order to prove the lemma, we can write
\begin{align}
&\hspace{-20pt}(K-M-P)L = H(W_{[M+P+1:K]}) \notag\\
&\stackrel{}{=} H(W_{[M+P+1:K]} | Q, W_{[1:M+P]} ) \notag\\
&\stackrel{\text{(a)}}{=} H(W_{[M+P+1:K]} | Q, W_{[1:M+P]} ) \notag\\
&\quad - H(W_{[M+P+1:K]} | A^{\mc{S},[K-P+1:K]}_{[1:N]}, Q, W_{[1:M+P]} ) \notag\\
&\stackrel{}{=} I(W_{[M+P+1:K]} ; A^{\mc{S},[K-P+1:K]}_{[1:N]} | Q, W_{[1:M+P]} ) \notag\\
&\stackrel{\text{(b)}}{=} H(A^{\mc{S},[K-P+1:K]}_{[1:N]} | Q, W_{[1:M+P]} ) \notag\\
&\stackrel{\text{}}{\le}  \sum_{n=1}^N H(A^{\mc{S},[K-P+1:K]}_n | Q, W_{[1:M+P]}) \notag\\
&\stackrel{\text{(c)}}{=} NH(A_1 | Q, W_{[1:M+P]})\notag\\
&= NH(A_1 | Q, W_\mc{S}, W_\mc{P})\notag,
\end{align}
where (a) follows since $P\ge\frac{K-M}{2}$, and
(b) holds since $A^{\mc{S},[K-P+1:K]}_{[1:N]}$ is a deterministic function of $Q$, $W_{[1:M+P]}$, and $W_{[M+P+1:K]}$, and (c) follows from \cite[Lemma~1]{Banawan2017}. This completes the proof of the lemma.
\end{proof}

\begin{proof}[{Proof of Lemma~\ref{lem:ConvLowP_lem1}}]
Without loss of generality we define $\mc{S}=\{1,\ldots,M\}$, $\mc{P}_1 = \{M+1,\ldots,M+P\}$ and $\mc{P}_2=\{M+P+1,\ldots,M+2P\}$.
Then, starting from \eqref{eq:eq99}, we can write
\begin{align}
&\hspace{-5pt} (K-M)L \stackrel{}{=} H(A_1, A^{\mc{S},\mc{P}_1}_{[2:N]}, \ldots, A^{\mc{S},\mc{P}_\beta}_{[2:N]} | Q, W_\mc{S}) \notag\\
&= H( A_1, A^{\mc{S},\mc{P}_1}_{[2:N]} | Q, W_\mc{S} ) + H( A^{\mc{S},\mc{P}_2}_{[2:N]} |  A_1, A^{\mc{S},\mc{P}_1}_{[2:N]}, Q, W_\mc{S} ) \notag\\
&\quad + H(A^{\mc{S},\mc{P}_3}_{[2:N]}, \ldots, A^{\mc{S},\mc{P}_\beta}_{[2:N]} | A_1, A^{\mc{S},\mc{P}_1}_{[2:N]}, A^{\mc{S},\mc{P}_2}_{[2:N]}, Q, W_\mc{S}) \notag\\
& \stackrel{\text{(a)}}{\le} N H(A_1 | Q, W_\mc{S})  \notag\\
&\quad + H( A^{\mc{S},\mc{P}_2}_{[2:N]} |  A_1, A^{\mc{S},\mc{P}_1}_{[2:N]}, W_{[M+1:M+P]}, Q, W_\mc{S} ) \notag\\
&\quad + H(A^{\mc{S},\mc{P}_3}_{[2:N]}, \ldots, A^{\mc{S},\mc{P}_\beta}_{[2:N]} | A_1, A^{\mc{S},\mc{P}_1}_{[2:N]}, A^{\mc{S},\mc{P}_2}_{[2:N]}, \notag\\
&\hspace{140pt} W_{[M+1:M+2P]}, Q, W_\mc{S}) \notag\\
& \stackrel{\text{(b)}}{\le} N H(A_1 | Q, W_\mc{S})  +  H( A^{\mc{S},\mc{P}_2}_{[2:N]} |  W_{[M+1:M+P]}, Q, W_\mc{S} ) \notag\\
&\quad + H(A^{\mc{S},\mc{P}_3}_{[2:N]}, \ldots, A^{\mc{S},\mc{P}_\beta}_{[2:N]} | A_1, W_{[M+1:M+2P]}, Q, W_\mc{S}) \notag\\
& \stackrel{}{\le} N H(A_1 | Q, W_\mc{S})  +  H( A^{\mc{S},\mc{P}_2}_{[2:N]} |  W_{[M+1:M+P]}, Q, W_\mc{S} ) \notag\\
&\quad + H(A^{\mc{S},\mc{P}_3}_{[1:N]}, \ldots, A^{\mc{S},\mc{P}_\beta}_{[1:N]} |  W_{[M+1:M+2P]}, Q, W_\mc{S} ) \notag\\
&\quad - H(A_1 | W_{[M+1:M+2P]}, Q, W_\mc{S} ) \notag\\
& \stackrel{\text{(c)}}{\le} N H(A_1 | Q, W_\mc{S})  +  H( A^{\mc{S},\mc{P}_2}_{[2:N]} |  W_{[M+1:M+P]}, Q, W_\mc{S} ) \notag\\
&\quad + (K-M-2P)L - H(A_1 | W_{[M+1:M+2P]}, Q, W_\mc{S} ) \notag\\
&\stackrel{\text{(d)}}{\le} N H(A_1 | Q, W_\mc{S}) + (N-1)[NH(A_1 | Q, W_\mc{S}) - PL] \notag\\
&\quad + (K-M-2P)L - H(A_1 | W_{[M+1:M+2P]}, Q, W_\mc{S} ) \notag
\end{align}
where (a) follows from \cite[Lemma~1]{Banawan2017} and the decodability of $W_{[M+1:M+2P]}$ from $(A_1, A^{\mc{S},\mc{P}_1}_{[2:N]}, A^{\mc{S},\mc{P}_2}_{[2:N]},  W_\mc{S})$, (b) is true since conditioning does not increase entropy, (c) follows since messages $W_{M+2P+1},\ldots,W_{K}$ are appeared in one of $\mc{P}_3,\ldots,\mc{P}_\beta$, where (d) follows from Lemma~\ref{lem:InterferenceCondLemma}.	
\end{proof}

\begin{lemma}\label{lem:InterferenceCondLemma}
	Consider $\mc{P}_1,\mc{P}_2\subseteq [1:K]\setminus \mc{S}$, $|\mc{P}_1| = |\mc{P}_2| = P$, and $\mc{P}_1\cap\mc{P}_2 = \varnothing$, then we have
	\[
	H( A^{\mc{S},\mc{P}_2}_{[2:N]} | W_{\mc{P}_1}, Q, W_{\mc{S}} ) \le (N-1) [N H(A_1 | Q, W_{\mc{S}})-PL].	
	\]
\end{lemma}
\begin{proof}[{Proof of Lemma~\ref{lem:InterferenceCondLemma}}]
We can write
\begin{align}
&\hspace{-10pt} H( A^{\mc{S},\mc{P}_2}_{[2:N]} | W_{\mc{P}_1}, Q, W_{\mc{S}} ) \notag\\
&\le \sum_{n=2}^N H( A^{\mc{S},\mc{P}_2}_n | W_{\mc{P}_1}, Q, W_{\mc{S}} ) \notag\\
&\stackrel{}{\le} \sum_{n=2}^N H(  A^{\mc{S},\mc{P}_1}_{[1:n-1]}, A^{\mc{S},\mc{P}_2}_n,  A^{\mc{S},\mc{P}_1}_{[n+1:N]} | W_{\mc{P}_1}, Q, W_{\mc{S}} ) \notag\\
&\stackrel{}{=} \sum_{n=2}^N \Big[ H(  A^{\mc{S},\mc{P}_1}_{[1:n-1]}, A^{\mc{S},\mc{P}_2}_n,  A^{\mc{S},\mc{P}_1}_{[n+1:N]}, W_{\mc{P}_1}  | Q, W_{\mc{S}} ) \notag\\
&\hspace{160pt} - H(W_{\mc{P}_1}  | Q, W_{\mc{S}} ) \Big] \notag\\
&\stackrel{}{=} \sum_{n=2}^N \Big[ H(  A^{\mc{S},\mc{P}_1}_{[1:n-1]}, A^{\mc{S},\mc{P}_2}_n,  A^{\mc{S},\mc{P}_1}_{[n+1:N]}  | Q, W_{\mc{S}} ) \notag\\
&\quad \quad\quad + H( W_{\mc{P}_1} | A^{\mc{S},\mc{P}_1}_{[1:n-1]}, A^{\mc{S},\mc{P}_2}_n,  A^{\mc{S},\mc{P}_1}_{[n+1:N]}, Q, W_{\mc{S}} ) \notag\\
&\hspace{160pt}- H(W_{\mc{P}_1} | W_{\mc{S}} ) \Big] \notag\\
&\stackrel{\text{(a)}}{=} \sum_{n=2}^N \Big[ H(  A^{\mc{S},\mc{P}_1}_{[1:n-1]}, A^{\mc{S},\mc{P}_2}_n,  A^{\mc{S},\mc{P}_1}_{[n+1:N]}  | Q, W_{\mc{S}} ) \notag\\
&\hspace{160pt} - H(W_{\mc{P}_1}  |  W_{\mc{S}} ) \Big] \notag\\
&\stackrel{\text{(b)}}{\le} \sum_{n=2}^N NH(A_1 | Q, W_\mc{S}) - H(W_{\mc{P}_1}  | W_{\mc{S}} ) \notag\\
&\stackrel{}{=} (N-1)N H(A_1 | Q, W_\mc{S}) - PL
\end{align}
where (a) follows from the privacy constraint stated in \cite[Section~6, Page~33]{Banawan2017}, \ie, $W_{\mc{P}_1}$ can be decoded from $\left( A^{\mc{S},\mc{P}_1}_{[1:n-1]}, A^{\mc{S},\mc{P}_2}_n,  A^{\mc{S},\mc{P}_1}_{[n+1:N]}, Q, W_{\mc{S}} \right)$, and (b) follows from \cite[Lemma~1]{Banawan2017}. This completes the proof of the lemma.
\end{proof}

\end{document}